\documentclass[a4paper,12pt,reqno,superscriptaddress]{revtex4}

\usepackage[centertags]{amsmath}
\usepackage{amsfonts}
\usepackage{amssymb}
\usepackage{amsthm}
\usepackage{newlfont}
\usepackage{stmaryrd}
\usepackage{mathrsfs}
\usepackage{euscript}
\usepackage{graphicx}

\usepackage{wrapfig}
\usepackage{subfigure}
\usepackage{amscd}

\theoremstyle{plain}
  \newtheorem{theorem}{Theorem}[section]
  \newtheorem{corollary}[theorem]{Corollary}
  \newtheorem{proposition}[theorem]{Proposition}
  
\theoremstyle{definition}

\theoremstyle{remark}

  \newtheorem{example}[theorem]{Example}

 \let\be=\beta \let\de=\delta

\let\De=\Delta \let\Ga=\Gamma  
\let\Th=\Theta

\newcommand{\caD}{{\mathcal D}}

\newcommand{\caG}{{\mathcal G}}

\newcommand{\caO}{{\mathcal O}}

\newcommand{\caT}{{\mathcal T}}

\newcommand{\bbR}{{\mathbb R}}

\newcommand{\bbZ}{{\mathbb Z}}

\newcommand{\opunit}{\text{1}\kern-0.22em\text{l}}

\DeclareMathAlphabet{\mathpzc}{OT1}{pzc}{m}{it}

\begin{document}

\title{Heat bounds and the blowtorch theorem}

\author{Christian Maes}
\email{christian.maes@kuleuven.ac.be}
\affiliation{Instituut voor Theoretische Fysica, KU Leuven}
\author{Karel Neto\v{c}n\'{y}}
\email{netocny@fzu.cz}
\affiliation{Institute of Physics, Academy of Science, Czech Republic}

\begin{abstract}
We study driven systems with possible population inversion and we give optimal bounds on the relative occupations in terms of released heat.  A precise meaning to Landauer's blowtorch theorem (1975) is obtained stating that nonequilibrium occupations are essentially modified by kinetic effects.  Towards very low temperatures we apply a Freidlin-Wentzel type analysis for continuous time Markov jump processes.  It leads to a definition of dominant states in terms of both heat and escape rates.
\end{abstract}


\maketitle
\section{Stochastic network}
The construction of nonequilibrium statistical mechanics requires to develop general concepts and mathematical techniques for dealing with driven systems.  Useful inspiration can already be obtained from the simplest systems undergoing a Markov evolution.  A well-known example is Schnakenberg's network theory where affinities are defined from the cycles of the graph associated with a Markov jump process~\cite{Sch,qia}. A much more recent example maps nonequilibrium fluxes onto a Markov process on a cycle space~\cite{alt}.  Still far from treating phase transitions, we feel that a nonequilibrium theory must include extensions of the Gibbs formalism where
besides heat and energy, also other kinetic aspects are crucial for the evaluation of state occupations.  The present paper provides rigorous information on that question, first by giving new heat bounds on the stationary occupations in Section \ref{sec:heatbounds} and continuing in Section \ref{blowsec} by mathematically stating in the framework of Markov jump processes an insight of Landauer, known as the blowtorch theorem \cite{land}.  We end in Section~\ref{lowt} with a low temperature analysis to develop the notion of dominant states.  There we meet most prominently the two contributions, one referring to the ``life-time'' of states and the other to their ``accessibility'', that determine the (asymptotic) stationary occupations outside equilibrium.  The main method is to represent the Markov jump process as a stochastic network and to apply a flow-graph or Kirchhoff formula for the stationary distribution.  This formula is well-known and was also used in the evaluation of stochastic stability in spatial games, see e.g.~\cite{schu,jac}.\\

We consider a graph  $\caG$ with nodes $x,y,\ldots \in K$ that represent physical states on some reduced level of description such as the collection of particle positions or spins or chemo-mechanical configurations of molecules.   Between any two states are associated the
heat $q(x,y) = -q(y,x) \in \bbR$ and the activation $\psi(x,y) = \psi(y,x) \geq 0$, where the words refer to their physical meaning in an Arrhenius-formula for transition rates $k(x,y)$.  More specifically we introduce the Markov jump process with rates
\begin{equation}\label{tra}
k(x,y) := \psi(x,y)\,\exp \Bigl[ \frac{\beta}{2} q(x,y) \Bigr]
\end{equation}
where $\beta\geq 0$ is called the inverse temperature of the environment.  In that way, the process satisfies the local detailed balance condition, \cite{kls,hal}
\begin{equation}\label{ldb}
  \log\frac{k(x,y)}{k(y,x)} = \be q(x,y)
\end{equation}
with $\beta q(x,y)$ the entropy flux under the jump $x\rightarrow y$.
On the other hand, the escape rates $\xi(x) := \sum_y k(x,y)$  depend also on the prefactors $\psi(x,y)$.  The edges of $\caG$ are made by the pairs $(x,y)$ over which $\psi(x,y) > 0$.  To these values we also refer as determining the kinetics of the process.   Physically that also depends on the temperature,
$\psi(x,y) = \psi(x,y;\beta)$, but we will postpone to make that dependence more explicit till Section~\ref{lowt}.\\

For mathematical simplicity we assume that the state space $K$ is finite and that the Markov process defined above is irreducible (and hence also time-translation ergodic in its stationary regime).  The backward generator is
$(L g)(x) := \sum_y k(x,y)\,[g(y) - g(x)]$, and the unique stationary probability distribution $\rho>0$ on $K$ is characterized by  $\sum_x \rho(x) (Lg)(x) =0$ for all functions $g$.  That is equivalent to the stationary Master equation. $\sum_y [\rho(x)k(x,y) - \rho(y)k(y,x)] = 0$ for all $x$.  We also refer to $\rho(x)$ as the stationary occupations and we are interested to understand the behavior of these occupations in terms of the heat $\{q(x,y)\}$, activation $\{\psi(x,y)\}$ and inverse temperature $\beta$.\\

A useful representation of the stationary distribution is in terms of the Kirchhoff formula, \cite{schu,jac,bol}:
\begin{equation}\label{kirchhoff}
  \rho(x) = \frac{W(x)}{\sum_y W(y)}\,,\qquad
  W(x) := \sum_\caT w(\caT_x)
\end{equation}
in which the last sum runs over all spanning trees in the graph $\caG$ and $\caT_x$ denotes the in-tree to $x$ defined for any tree $\caT$ and state $x$ by orienting every edge in $\caT$ towards $x$; its weight $w(\caT_x)$ is
\begin{equation}\label{kirchhoff1}
  w(\caT_x) := \prod_{b \in \caT_x} k(b)
\end{equation}
i.e.,\ the product of transition rates $k(b) \equiv k(y,z)$ over all oriented edges
$b \equiv (y,z)$ in the in-tree $\caT_x$. (To simplify notation, we identify the graph with the set of its edges. We refer to \cite{bol} for the necessary elements of graph theory.)

\section{Heat bounds on the stationary occupations}\label{sec:heatbounds}

For any oriented path $\caD$ along the edges of the graph $\caG$ we write $q(\caD)$ for the total dissipated heat along $\caD$, i.e.,
$q(\caD) := \sum_b q(b)$. For any
spanning tree $\caT$ we also define $\caT_{xy}$ as the (unique) oriented path from $x$ to $y$ along the edges of $\caT$.

\begin{proposition}
For all states $x,y\in K$, under the stationary distribution $\rho$,
\begin{equation}\label{occupation-rel}
  \frac{\rho(x)}{\rho(y)} =
  \frac{\sum_{\caT} e^{-\be q(\caT_{xy})} w(\caT_y)}{\sum_{\caT} w(\caT_y)}
\end{equation}
where the sums are over all spanning trees in the graph $\caG$.
\end{proposition}
\begin{proof}

For any tree $\caT$ and states $x, y$ we decompose the corresponding in-trees (as sets of oriented edges) as
$\caT_x = \caT_{y x} \cup (\caT_x \setminus \caT_{y x})$ and
$\caT_y = \caT_{x y} \cup (\caT_y \setminus \caT_{x y})$, respectively. Observing that
(i) $\caT_{x y} = -\caT_{y x}$ (the orientation of all edges inverted) and (ii)
$\caT_x \setminus \caT_{y x} = \caT_y \setminus \caT_{x y}$ (the orientation of all edges kept unchanged), the local detailed balance~\eqref{ldb} yields
\begin{equation}
  \frac{w(\caT_x)}{w(\caT_y)} = e^{-\be q(\caT_{xy})}
\end{equation}
This being applied in~\eqref{kirchhoff}--\eqref{kirchhoff1} proves~\eqref{occupation-rel}.
\end{proof}

The representation~\eqref{occupation-rel} immediately implies
\begin{corollary}\label{cor:bounds}
The relative stationary occupations satisfy
\begin{equation}\label{occupation-bounds}
  \min_{y \stackrel{\!\!\caD}{\rightarrow x}} q(\caD) \leq
  \frac{1}{\be} \log\frac{\rho(x)}{\rho(y)} \leq
  \max_{y \stackrel{\!\!\caD}{\rightarrow x}} q(\caD)
\end{equation}
with the minimum and maximum taken over all oriented paths (i.e., self-avoiding walks) from $y$ to $x$
in $\caG$.  When either one of the two inequalities is an equality, then both inequalities become equalities and the minimum coincides with the maximum.
\end{corollary}
In words, the  relative occupations $\rho(x)/\rho(y)$ can be estimated via the heat along the most- and the least- dissipative paths connecting $x$ and $y$. Note that in the case of \emph{global} detailed balance $q(x,y) = E(x)- E(y)$ in \eqref{tra}--\eqref{ldb} for some energy function $E(x)$, and the heat $q(\caD) = E(x) - E(y)$ only depends on the initial and final configurations $x,y$ of the path $\caD$.
Then, the lower and upper bounds coincide for all pairs $x, y$ and~\eqref{occupation-bounds} reproduces the Boltzmann equilibrium statistics for $\rho$.

\section{Blowtorch theorem}\label{blowsec}
The heat function $q(x,y)$ provides a partial ordering in the following sense.
We write $x \succ y$ (respectively $x \succeq y$) whenever
$q(\caD) > 0$ (respectively $q(\caD) \geq 0$)
along all oriented paths $\caD$ from $y$ to $x$; one readily checks the transitivity condition for partial order. By Corollary~\ref{cor:bounds},
$x \succeq y$ implies $\rho(x) \geq \rho(y)$. Moreover, if
$x \succ y$ then $\rho(y)/\rho(x)$ converges to zero in the zero-temperature limit
$\be \to +\infty$, exponentially with rate bounded from below through the least-dissipation path from $y$ to $x$. Sometimes (typically in regimes not very far from equilibrium), this can be enough in order to determine the asymptotically most populated configuration(s) analogous to equilibrium ``ground states'', as well as to get bounds on the leading excitations. However, in general, considerations based only on heat may become insufficient, and the next section describes exactly what information is needed.

Both relations $x \succeq y$ and $y \succeq x$ are obeyed simultaneously only if $q(\caD) = 0$ for all paths connecting $x$ and $y$, in which case
$\rho(x) = \rho(y)$. On the other hand, any two configurations $x$ and $y$ are \emph{incomparable} in the sense that neither $x \succeq y$ nor $x \preceq y$, if and only if there exist two oriented paths $\caD_{1,2}$ from $x$ to $y$ such that $q(\caD_1) < 0 < q(\caD_2)$. In this case it cannot be decided which of the occupations
$\rho(x)$, $\rho(y)$ is larger on the basis of the heat functions $q(x,y)$ only and the symmetric components of the transition rates become essential.
In fact, any chosen path from $x$ to $y$ can be arbitrarily enhanced (or all the remaining ones arbitrarily suppressed) by suitably adjusting these symmetric parts, cf.~\eqref{occupation-rel},
indicating that either of the two inequalities between
$\rho(x)$ and $\rho(y)$ can indeed be attained.
This is the essence of the so called ``blowtorch theorem,'' heuristically introduced
by Landauer~\cite{land} to argue that entropy production (variational) principles cannot have a general validity. We can now formalize that as
\begin{proposition}[Blowtorch theorem]\label{pro:blowtorch}
Whenever for a pair of states $x^*$ and $y^*$ neither
$x^*\succeq y^*$ nor $y^* \succeq x^*$ holds, then without changing the heat function $\{q(x,y)\}$, we can  always make either $\rho(x^*) > \rho(y^*)$ or
$\rho(x^*) < \rho(y^*)$ by changing the kinetics $\{\psi(x,y)\}$.
\end{proposition}
\begin{proof}
Assuming the hypothesis is true, there exists a path $\caD$ from state $y^*$ to $x^*$ such that $q(\caD) > 0$. Let $\caT$ be a spanning tree such that
$\caT_{y^* x^*} = \caD$ and define
$\psi(b) = 1$ for all $b \in \caT$ and $\psi(b) = \de > 0$ otherwise.
Then for $\de$ small enough the tree $\caT$ provides a dominant contribution to the tree-sums in~\eqref{occupation-rel} such that
$\rho(x^*) / \rho(y^*) > e^{\be q(\caD)/2}  > 1$. The existence of $\psi$ such that
$\rho(x^*) / \rho(y^*) < 1$ follows by the same argument.
\end{proof}

\begin{example}\label{ex:walk}
Let us consider a random walk on the ring $\{1,2,\ldots,N\}$ where we identify $N+1=1$.  The transition rates are taken
\[
k(x,x+1) = p_x\,e^{\frac{\beta}{2} q_x},\quad k(x+1,x) = p_x\,e^{-\frac{\beta}{2} q_x},\qquad x\in \bbZ_N
\]
for $p_x>0$ and total heat $Q = \sum_{x=1}^N q_x$ along the ring.  Global detailed balance implies $Q=0$; here we assume $Q>0$.
Between any $x,y\in \bbZ_N$ there is a positively (clockwise) and a negatively (anti-clockwise) oriented path, $\caD^+(x,y)$ respectively $\caD^-(x,y)$, with corresponding heat along these paths satisfying
\[
q(\caD^+(x,y)) = -q(\caD^-(y,x)),\qquad q(\caD^+(x,y))+ q(\caD^+(y,x)) =  Q
\]
In particular, $q(\caD^+(x,x+1))=q_x$ and
$q(\caD^-(x,x+1)) = q_x - Q$.
Hence, the sites $x$ and $x+1$ can be heat-ordered ($\succeq$) provided
both heat quanta $q(\caD^\pm(x,x+1))$ have the same sign, i.e.\ for
$q_x \not\in (0,Q)$. If this condition is verified for all sites $x$ then the heat-order becomes \emph{complete}, extending the detailed balance case
($Q = 0$) where the completeness is obvious. Remark that whenever $q_x \geq 0$ for \emph{all} sites $x$ then the heat-order can only be \emph{partial} since our complete-order condition is then in contradiction with the constraint
$\sum_x q_x = Q > 0$. Furthermore, in a diffusion regime where $q_x = \caO(1/N)$ while keeping $Q = \caO(1)$, the heat-order only applies for the neighboring sites $(x,x+1)$ such that $q_x \leq 0$.

To illustrate the blowtorch theorem, Proposition~\ref{pro:blowtorch}, we take the specific case $N = 3$.
The stationary occupations follow~\eqref{occupation-rel},
\[
\frac{\rho(1)}{\rho(2)} = \frac{e^{-\beta q_1} \,p_1\, p_2 \,e^{\frac{\beta}{2} (q_1-q_2)} +
e^{\beta (q_2+q_3)} \,p_2\, p_3 \,e^{-\frac{\beta}{2} (q_2+q_3)} + e^{-\beta q_1} \,p_1\, p_3 \,e^{\frac{\beta}{2} (q_1+q_3)}}{p_1\, p_2 \,e^{\frac{\beta}{2} (q_1-q_2)} +  p_2\, p_3 \,e^{-\frac{\beta}{2} (q_2+q_3)} + p_1\, p_3 \,e^{\frac{\beta}{2} (q_1+q_3)}}
\]
\emph{Non-ambiguous case.}
It is immediate that when $q_2 + q_3 = -q_1$ (detailed balance), then
$\rho(1) = e^{-\beta q_1}\,\rho(2)$.  But we also see that if
$q(\caD^+(1,2)) = q_1\geq 0$ and $q(\caD^-(1,2)) = -q_3-q_2 \geq 0$, then still $\rho(1) \leq \rho(2)$, independent of the $p_{1,2,3}$. Here the heat released along both paths from $1$ to $2$ is always nonnegative and that is why the occupation of $1$ cannot exceed the occupation of $2$.\\
\emph{Ambiguous case.}
Take now both $q_1 > 0$ and $q_2 + q_3 > 0$ as well.  Then, by taking $p_1$ very small we have $\rho(1) > \rho(2)$ while for $p_3$ very small we find $\rho(1) < \rho(2)$.
In this case the heats along the two paths from $1$ to $2$ have a different sign, and hence the higher-occupied state cannot be determined without knowing the detailed kinetics $\psi$ (and both possibilities can indeed be designed).
\end{example}

If there are some \emph{a priori} energy levels
$E(1) < E(2) < \ldots$ assigned to the states so that the heat functions take the form
$q(x,y) = E(x) - E(y) + F(x,y)$ where $F$ accounts for the nonequilibrium driving, then the situation where, e.g., $\rho(1) < \rho(2)$ is usually called a population inversion. We have seen that a necessary condition for this effect to occur is the absence of heat-order between the states $1$ and $2$.

Of course the fact that the $\{\psi(x,y)\}$ \emph{also} determine the stationary distribution is rather obvious from the mathematical point of view.
What is in fact more strange is that at detailed balance (and in the neighborhood of it) heat and heat alone plays such a dominant role, and determines on its own the stationary occupations.  It is then the challenge of nonequilibrium physics to identify a complement to heat and entropy for characterizing the steady behavior. In other words, to understand exactly what aspect of the kinetics matter.  Such studies are under way for different aspects of nonequilibrium theory, including linear response relations and  dynamical fluctuation theory where the concept of activity, traffic or frenesy have been used to tag indeed the relevant kinetic features, e.g. in \cite{mnw,mw}.\\

The original Landauer's blowtorch example considers a particle moving over an energy barrier, both sides of which are being heated up differently~\cite{land}. That somewhat differs from our isothermal stochastic framework but we restrict to the latter to summarize Landauer's idea into a concise mathematical statement. Other results, more vaguely related to the original concept are also known, e.g., the possibility to achieve an arbitrary stationary occupation distribution $\mu >0$ by adding a potential $V_\mu = V$ to the heat function, $q(x,y) \mapsto q(x,y) + V(x) - V(y)$ while keeping the kinetics fixed, plays an essential role in the Donsker-Varadhan large deviation theory~\cite{dv}; the extra potential $V$ was also called a ``blowtorch'' in~\cite{mnw}.

\section{Low-temperature asymptotics}\label{lowt}

To get more insight into how different aspects of the dynamics enter the structure of the stationary distribution, we next look into its asymptotics at low temperatures. We apply a variant of Freidlin and Wentzell analysis for small noise~\cite{fw} adapted to continuous-time Markov chains. Mathematically there are no new ideas, but the continuous--time extension is relevant for physical applications. It adds the life--time of states as another component in the stationary occupations complementing the accessibility of states within the embedded discrete-time network. A possible frustration between both components then leads to a variety of low-temperature patterns.

We add the assumption that the transition rates $k(x,y) = k(x,y;\be)$ have the zero-temperature (logarithmic) limit
\begin{equation}\label{assumption-LD}
  \phi(x,y) := \lim_{\be \to +\infty} \frac{1}{\be} \log k(x,y;\be), \qquad x,y\in K
\end{equation}
abbreviated as $k(x,y) \asymp e^{\be \phi(x,y)}$. Since by local detailed balance~\eqref{ldb} we know that
$\phi(x,y) - \phi(y,x) = q(x,y)$ with the heat $q(x,y)$ temperature-independent, the added premise is thus that $\frac{1}{\beta} \log \psi(x,y;\beta)
\rightarrow \phi(x,y) - q(x,y)/2$ has a well defined limit when $\be \to +\infty$.
The escape rates have the asymptotics
$\xi(x) = \sum_y k(x,y) \asymp e^{-\be \Ga(x)}$ with
$\Ga(x) := -\max_y \phi(x,y)$ specifying the typical life-time of state $x$ as
$\tau(x) \simeq e^{\be\Ga(x)}$. It is useful to decompose
$-\phi(x,y) =: \Ga(x) + U(x,y)$ defining $U(x,y) \geq 0$.
\begin{proposition}
The stationary distribution has the exponential asymptotics
\begin{equation}\label{stab}
  \lim_{\be \to +\infty} \frac{1}{\be} \log \rho(x) = \tilde\Psi(x) - \max_x \tilde\Psi(x)
  \end{equation}
  with
\begin{equation}\label{stab1}
\tilde\Psi(x)  := \Ga(x) - \min_\caT U(\caT_x)
\end{equation}
for  $U(\caT_x) := \sum_{(y,z) \in \caT_x} U(y,z)$.
\end{proposition}

\begin{proof}
To find the stationary distribution we apply the Kirchhoff formula \eqref{kirchhoff}.
We evaluate
\begin{equation}
\begin{split}
\lim_{\be\to +\infty} \frac{1}{\be}
  \log W(x) &= \lim_{\be\to +\infty} \frac{1}{\be}
  \log\sum_{\caT} \prod_{(y,z) \in \caT_x} e^{\be \phi(y,z)}\\
  &= \lim_{\be\to +\infty} \frac{1}{\be}
  \log\sum_{\caT} \Bigl( \prod_{y \neq x} e^{-\be\Ga(y)} \Bigr)
  \prod_{(y,z) \in \caT_x} e^{-\be U(y,z)}
\\
  &= -\sum_{y \neq x} \Gamma(y) - \min_\caT U(\caT_x)
\\
  &= \tilde\Psi(x) - \sum_y \Ga(y)
\end{split}
\end{equation}
The normalization in~\eqref{kirchhoff} contributes the largest $W(y)$, which ends the proof
of formula \eqref{stab}.
\end{proof}
Other representations are possible 
and depending on the physical context, they may become more natural. As example, for a network made of metastable states well-separated by energy barriers, we could write
$\phi(x,y) = E(x) - \Delta(x,y)$ where $E(x)$ can be interpreted as the energy of state $x$ and $\Delta(x,y)$ as an effective energy barrier between $x$ and $y$.  At \emph{global} detailed balance $\Delta(x,y) = \Delta(y,x)$,
and under \emph{local} detailed balance
$\Delta(x,y) - \Delta(y,x)$ corresponds to the work of additional non--conservative forces in the transition $x\rightarrow y$.
Formula \eqref{stab} can now be written as
\[
\rho(x) \asymp e^{\beta[\Omega - E(x) - \Theta(x)]}
\]
where
\[
\Theta(x) := \min_{\caT}\sum_{(y,z)\in \caT_x}\Delta(y,z), \quad \Omega := \min_x[E(x) + \Theta(x)]
\]
At global detailed balance (symmetric $\De$), $\Th(x)$ becomes a constant determined by the
minimizing spanning tree and we recover the Boltzmann distribution.\\

In the asymptotic regime $\beta\to +\infty$, some transitions are typical whereas the others become exponentially damped. Indeed, when the system makes a jump from $x$ then the probability to go to a state $y$ asymptotically goes like
$p(x,y) = k(x,y) / \xi(x) \asymp e^{-\be U(x,y)}$. For any $x$ there is always at least one other state $y$ such that $U(x,y) = 0$ --- we call these states \emph{preferred successors of $x$} and we consider the digraph $\caG^D$ made of all states and directed arcs indicating all preferred successors. Clearly, those transitions which are in $\caG$ but do not correspond to a directed arc in $\caG^D$ become suppressed at low temperatures.  Note that $\caG^D$ may not be (even weakly)
connected and that it may contain both arcs $(x,y)$ and $(y,x)$.\\

Let us denote
$\Psi(x) := \tilde\Psi(x) - \max_y \tilde\Psi(y)$, i.e.,
$\rho(x) \asymp e^{\be\Psi(x)}$.
By construction,
$\Psi(x) \leq 0$ and there is always at least one state $x^*$ such that
$\Psi(x^*) = 0$.
Adopting the terminology of~\cite{jac},
we call these states \emph{dominant}: the stationary occupation of the set of all dominant states is exponentially close to unity in the limit
$\be \to +\infty$.
That resembles the problem of stochastic stability but note that we do not have in general a unique and well-defined zero temperature dynamics.
Dominant states are the analogue of (dominant) ground states for equilibrium statistical models, \cite{bs}. Yet, here we do not minimize the (free) energy;
the dominance of a state follows from its long life-time (the first term on the right-hand side of~\eqref{stab1}) and its accessibility from other states (the second term).
For a related analysis in terms of radius and coradius, see~\cite{ell}.

The easiest situation appears when there is no frustration between both terms: a state
$x^*$ is called \emph{absolutely dominant} if (i) it has a maximal life-time among all states, $\Ga(x^*) = \max_y \Ga(y)$, and (ii) there exists an in-tree
$\caT_{x^*} \subset \caG^D$ (i.e., a tree which makes the second term to vanish). Clearly, if there exists an absolutely dominant state then \emph{all} dominant states are absolutely dominant.

Let us finish with an important special case. Assume that the digraph $\caG^D$ is strongly connected, i.e., for any two states $x, y$ there exist both paths from $x$ to $y$ and from $y$ to $x$ along the directed edges of $\caG^D$. In such an ``Escherian world''~\cite{es}, for \emph{every} state $x$ there exists an in-tree
$\caT_x \subset \caG^D$ and hence
$\tilde\Psi(x) = \Ga(x)$, i.e.,
$\rho(x) \asymp \exp\,[\be(\Ga(x) - \max_y \Ga(y))]$. In this case the stationary occupations are entirely determined by the life-times alone, the least frenetic ones. Informally, that is due to the absence of barriers which makes all states `equally accessible from everywhere'. Those states with maximal life-time are then absolutely dominant.
We note that this ``Escherian'' case need not contradict detailed balance, though these notions can only meet under certain degeneracy conditions: Since the former implies
$\rho(x)\,k(x,y) \asymp \exp[-\be(U(x,y) + \max_z \Ga(z))]$, global detailed balance 
requires 
$U(x,y) = U(y,x)$ and hence $(x,y) \in \caG^D$ whenever
$(y,x) \in \caG^D$. For $|K| > 2$ it requires that states with multiple preferred successors exist.\\
Continuing Example~\ref{ex:walk}, let us assume $q_x > 0$ and
$p_x \asymp 1$ (i.e., with essentially temperature-independent activation factors) for all $x = 1,\ldots, N$. That corresponds to
$\phi(x,x+1) = q_x/2$, $\phi(x,x-1) = -q_{x-1}/2 $, and hence
$\Ga(x) = -q_x/2$ and $U(x,x+1) = 0$, $U(x,x-1) = (q_x + q_{x-1})/2 > 0$. This is a simple ``Escherian'' case with the strongly connected digraph
$\caG^D = \{(x,x+1) \text{ for all $x$}\}$ and the asymptotic stationary occupation
$\rho(x) \asymp \exp(-\be q_x / 2 + \text{const})$. The states $x$ minimizing $q_x$ over the circle are absolutely dominant.

\section{Conclusions}

We have given bounds on the relative stationary occupations for Markov jump processes in terms of the released heat over paths in the Kirchhoff graph.  These heat bounds are optimal as they either resolve which of the states are more plausible, or it cannot be decided at all from heat alone and  other (kinetic) aspects of the dynamics essentially matter.  That has lead to a precisely stated blowtorch theorem. Finally, we have given the exponential low-temperature asymptotics of the occupations of states through their life-time and their reachability within the network. We expect such a representation in terms of dominant states and their excitations to be of help in the theory of strongly nonequilibrium processes under low-temperature conditions.

\begin{acknowledgements}
We are grateful to Jacek Miekisz for pointing out the work in \cite{jac} and in \cite{ell}.
KN acknowledges the support from the Grant Agency of the Czech Republic, Grant no.~P204/12/0897.
\end{acknowledgements}


\end{document}